\documentclass[final]{dmtcs-episciences}
\usepackage[utf8]{inputenc}
\usepackage{subfigure}
\usepackage[round]{natbib}
\usepackage{amsmath,amssymb,amsthm}
\usepackage{thmtools}
\usepackage{thm-restate}
\usepackage{complexity} 

\newcommand{\calD}{\mathcal{D}}
\newcommand{\calF}{\mathcal{F}}
\newcommand{\calW}{\mathcal{W}}
\newcommand{\true}{\texttt{True}}
\newcommand{\false}{\texttt{False}}
\newcommand{\yes}{\textsc{Yes}}

\newtheorem{theorem}{Theorem}
\newtheorem{lemma}{Lemma}

\newtheorem{observation}{Observation}

\newtheorem{reduction rule}{Reduction Rule}[section]
\newtheorem{marking-scheme}{Marking Scheme}[section]
\newtheorem{definition}{Definition}[section]

\author[Krithika et al.]{R. Krithika\affiliationmark{1}\thanks{The author is supported by SERB MATRICS grant number MTR/2022/000306.} \and Roohani Sharma\affiliationmark{2} \and Prafullkumar Tale\affiliationmark{3}\thanks{The author is supported by the INSPIRE Faculty Fellowship awarded by DST Govt of India.}}  

\title{The Complexity of Contracting Bipartite Graphs into Small Cycles} %

\affiliation{
 Indian Institute of Technology Palakkad, Palakkad, India\\
 Discrete Mathematics Group, Institute of Basic Science, Daejeon, South Korea\\
 Indian Institute of Science Education and Research Pune, Pune, India 
 }

\begin{document}
\publicationdata{vol. 27:3}{2025}{27}{10.46298/dmtcs.14658}{2024-11-01; 2024-11-01; 2025-10-15}{2025-11-23}

\maketitle
\begin{abstract}
For a positive integer $\ell \geq 3$, the $C_\ell$-\textsc{Contractibility} problem takes as input an undirected simple graph $G$ and determines whether $G$ can be transformed into a graph isomorphic to $C_\ell$ (the induced cycle on $\ell$ vertices) using only edge contractions. Brouwer and Veldman (1987) showed that $C_4$-\textsc{Contractibility} is \NP-complete\ in general graphs. It is easy to verify that $C_3$-\textsc{Contractibility} is polynomial-time solvable. Dabrowski and Paulusma (2017) showed that $C_{\ell}$-\textsc{Contractibility} is \NP-complete\ on bipartite graphs for $\ell = 6$ and posed as open problems the 
status of the problem when $\ell$ is 4 or 5. In this paper, we show that both $C_5$-\textsc{Contractibility} 
and $C_4$-\textsc{Contractibility} are \NP-complete\ on bipartite graphs. 
\end{abstract}

\section{Introduction}
Operations on graphs produce new graphs from existing ones. Elementary  editing operations include deleting vertices, deleting and/or adding edges, subdividing edges and contracting edges. Due to the ubiquitous presence of graphs in modelling real-world networks, many problems of practical importance may be posed as editing problems on graphs. In this work, we focus on modifying a graph by only performing edge contractions. Contracting an edge in a graph results in the addition of a new vertex adjacent to the neighbors of its endpoints followed by the deletion of the endpoints. As graphs typically represent binary relationships among a collection of objects, edge contractions naturally correspond to merging two objects into a single entity or to treating two objects as indistinguishable. Contractions can therefore be seen as a way of `simplifying' the graph and they have applications in clustering, compression, sparsification and computer graphics (\cite{AnderssonGL09,BernsteinDDKMS19,ChengDP04,CongL04,HarelK01,KarypisK98}). Edge contractions also play an important role in Hamiltonian graph theory, planar graph theory and graph minor theory (\cite{BrouwerV87,HoedeV81,RobertsonS95b}). 

For a collection $\mathcal{F}$ of graphs, the $\mathcal{F}$-\textsc{Contraction} problem takes a graph $G$ and an integer $k$ as input and decides if $G$ can be transformed into a graph that is isomorphic to a graph in $\mathcal{F}$ using at most $k$ edge contractions. Early works by \cite{WatanabeAN81,WatanabeAN83} and \cite{AsanoH83} showed that {\sc ${\cal F}$-Contraction} is \NP-complete\ even for well-structured graph classes $\mathcal{F}$ such as planar graphs, trees and chordal graphs. Asano and Hirata also showed that $\mathcal{F}$-\textsc{Contraction} remains \NP-complete\ on planar graphs where $\mathcal{F}$ is the set of trees, chordal graphs or outerplanar graphs. \cite{BrouwerV87} proved that \textsc{$\calF$-Contraction} is \NP-complete\ even when $\mathcal{F}$ is a singleton set consisting of a small graph like a cycle or a path on four vertices. This brings us to the \textsc{Graph Contractibility} problem~\cite[GT51]{GareyJ79}.

Given graphs $G$ and $H$, the \textsc{Graph Contractibility} problem decides whether $G$ can be transformed into a graph isomorphic to $H$ using only edge contractions. When $H$ is a fixed graph, the \textsc{Graph Contractibility} problem is called $H$-\textsc{Contractibility}. Intuitively, this problem of determining whether $G$ is contractible to $H$ may be seen as the task of determining if the `underlying structure' of $G$ is $H$. One of the related graph parameters in this context is cyclicity. The cyclicity of a graph is the largest integer $\ell$ for which the graph is contractible to the induced cycle on $\ell$ vertices (denoted as $C_\ell$). This parameter was introduced in the study of another important graph invariant called circularity by \cite{Blum82}. Ever since, there have been efforts towards understanding the complexity of computing cyclicity and expressing it in terms of some structural property of the graph. \cite{BrouwerV87} showed that $C_4$-\textsc{Contractibility} is \NP-complete, hence proving that determining cyclicity is \NP-hard\ in general. This result led to the study of the problem on special graph classes including bipartite graphs, claw-free graphs and planar graphs (\cite{DabrowskiP17-IPL,FialaKP13,Hammack99}). 

\cite{Hammack99} showed that the cyclicity of planar graphs can be computed in polynomial time and in another work \cite{Hammack02} showed that $C_{\ell}$-\textsc{Contractibility} is \NP-complete\ for every $\ell \ge 5$ in general. Later, \cite{KaminskiPT10} showed that $H$-\textsc{Contractibility} is polynomial-time solvable on planar graphs for any $H$. \cite{LevinPW08} showed that $H$-\textsc{Contractibility} is polynomial-time solvable on general graphs if $H$ is a graph on at most 5 vertices containing a universal vertex. 
However, \cite{HofKPST12} showed that the presence of a universal vertex in $H$ (on more than 5 vertices) does not guarantee that the $H$-\textsc{Contractibility} can be solved in polynomial time.
\cite{FialaKP13} showed that $C_{\ell}$-\textsc{Contractibility} is \NP-complete\ for claw-free graphs for every $\ell \ge 6$. \cite{HeggernesHLP14} proved that $P_\ell$-\textsc{Contractibility} is polynomial-time solvable on chordal graphs for every $\ell \geq 1$, where $P_{\ell}$ denotes the induced path on $\ell$ vertices. Later, \cite{BelmonteGHHKP14} proved that $H$-\textsc{Contractibility} is polynomial-time solvable on chordal graphs for every $H$. \cite{DabrowskiP17-IPL} showed that $C_6$-\textsc{Contractibility} is \NP-complete\ for bipartite graphs. 
It is easy to verify that that $C_3$-\textsc{Contractibility} is polynomial-time solvable in general graphs. 
In this paper, we show that both $C_5$-\textsc{Contractibility} and $C_4$-\textsc{Contractibility} are \NP-complete\ on bipartite graphs.

\begin{theorem}
\label{thm:c-5-contractiblity-np-hard}
$C_5$-\textsc{Contractibility} is \NP-complete\ on bipartite graphs.
\end{theorem}

\begin{theorem}
\label{thm:c-4-contractiblity-np-hard}
$C_4$-\textsc{Contractibility} is \NP-complete\ on bipartite graphs.
\end{theorem}

Theorems \ref{thm:c-5-contractiblity-np-hard} and \ref{thm:c-4-contractiblity-np-hard} involve reductions from the \textsc{Positive Not All Equal SAT (Positive NAE-SAT)} problem where given a formula $\psi$ in conjunctive normal form with no negative literals, the objective is to determine if there is an assignment of \true\ or \false\ to each of the variables such that for each clause at least one but not all variables in it are set to \true. Such an assignment is called a {\em not-all-equal satisfying assignment}. \textsc{Positive NAE-SAT} (also referred to as \textsc{Monotone NAE-SAT}) was shown to be \NP-complete\ by \cite{Schaefer78}. Also, a straight-forward reduction from \textsc{Set Splitting} or \textsc{Hypergraph $2$-Colorability} (\cite[SP4]{GareyJ79}) to \textsc{Positive NAE-SAT} ascertains this fact. \\

\section{Preliminaries}
\label{sec:preliminaries}
For a positive integer $q$, $[q]$ denotes the set $\{1, 2, \dots, q\}$. $\mathbb{N}$ denotes the collection of all non-negative integers. A partition of a set $S$ is a set of disjoint subsets of $S$ whose union is $S$. 

For standard graph-theoretic terminology not stated here, we refer the reader to the book by  \cite{Diestel12}. In this work, we only consider simple undirected graphs. Unless otherwise specified, we use $n$ to denote the number of vertices in the graph under consideration $G$. For an undirected graph $G$, its sets of vertices and edges, are denoted by $V(G)$ and $E(G)$, respectively. An edge between vertices $u$ and $v$ is denoted as $uv$. Two vertices $u, v$ in $V(G)$ are \emph{adjacent} if there is an edge $uv$ in $G$. The \emph{open neighborhood} of a vertex $v$, denoted by $N_G(v)$, is the set of vertices adjacent to $v$ and the \emph{closed neighborhood} of $v$, denoted by $N_G[v]$, is $N_G(v) \cup \{v\}$. A vertex $u$ is a \emph{pendant vertex} if $|N_G(v)| = 1$. The notion of neighborhood is extended to a set $S \subseteq V(G)$ of vertices by defining $N_G[S]$ as $\bigcup_{v \in S} N[v]$ and $N_G(S)$ as $N[S] \setminus S$. We omit the subscript in the notation for neighborhood if the graph under consideration is clear. 

A set $S  \subseteq V(G)$ of vertices is a \emph{dominating set} if $V(G) = N[S]$.  For a subset $F$ of edges, $V(F)$ denotes the set of endpoints of edges in $F$. For a subset $S$ of $V(G)$ ({resp.} a subset $F$ of $E(G)$), $G - S$ ({resp.} by $G - F$) denotes the graph obtained by deleting $S$ ({resp.} deleting $F$) from $G$. The subgraph of $G$ induced on the set $S \subseteq V(G)$ is denoted by $G[S]$. For two subsets $S_1, S_2$ of $V(G)$, $E(S_1, S_2)$ denotes the set of edges with one endpoint in $S_1$ and the other endpoint in $S_2$. With a slight abuse of notation, we use $E(S)$ to denote $E(S, S)$. We say thats the sets $S_1, S_2$ are adjacent if $E(S_1, S_2) \neq \emptyset$. 

A {\em path} $P$ in $G$ is a sequence $(v_1,\dots,v_k)$ of distinct vertices such that for each $i \in [k-1]$, $v_iv_{i+1} \in E(G)$. A {\em cycle} $C$ in $G$ is a sequence $(v_1,\dots,v_k)$ of distinct vertices such that $(v_1,\dots,v_k)$ is a path and $v_kv_1 \in E(G)$. A cycle $C=(v_1,\dots,v_k)$ is called an {\em induced} (or {\em chordless}) cycle if there is no edge in $G$ that is between two non-consecutive vertices of $C$ with the exception of the edge $v_kv_1$. The length of a path or cycle $X$ is the number of vertices in it and is denoted by $|X|$. An induced cycle of length $q$ is called a $q$-cycle and denoted by $C_q$. The \emph{distance} between any two vertices $u, v$ in $V(G)$ is the length of a shortest path from $u$ to $v$ in $G$. The \emph{diameter} of $G$ is the maximum length of a shortest path between two vertices in $G$. A graph is {\em connected} if there is a path between every pair of distinct vertices. A subset $S$ of $V(G)$ is said to be a \emph{connected set} if $G[S]$ is connected. A \emph{spanning tree} of a connected graph is a connected acyclic subgraph which includes all the vertices of the graph. A \emph{spanning forest} of a disconnected graph is a collection of spanning trees of its components. 

A set of vertices $Y$ is said to be an \emph{independent set} if no two vertices in $Y$ are adjacent. A graph $G$ is a \emph{bipartite graph} if its vertex set can be partitioned into two sets $X$ and $Y$ such that every edge in the graph has one endpoint in $X$ and the other endpoint in $Y$. Such a partition $\{X, Y\}$ of a bipartite graph is called a bipartition. The \emph{subdivision} of the edge $uv$ in $G$ results in another graph that is obtained from $G$ by deleting the edge $uv$ and adding a new vertex $w$ adjacent to $u$ and $v$. Observe that subdividing all edges of an arbitrary graph results in a bipartite graph. A complete bipartite graph with bipartition $\{X, Y\}$ is a bipartite graph where every vertex of $X$ is adjacent to every vertex of $Y$.

The {\em contraction} of an edge $e=uv$ in $G$ results in another graph denoted by $G/e$ that is obtained from $G$ by deleting vertices $u$ and $v$ from $G$, and adding a new vertex which is adjacent to the vertices that are adjacent to either $u$ or $v$ in $G$. This process does not introduce self-loops or parallel edges. Formally $G/e$ is defined as $V(G/e) = (V(G) \cup \{w\}) \backslash \{u, v\}$ and $E(G/e) = \{xy \mid x,y \in V(G) \setminus \{u, v\}, xy \in E(G)\} \cup \{wx \mid x \in N_G(u) \cup N_G(v)\}$ where $w$ is a new vertex. Observe that contracting an edge reduces the number of vertices in the graph by exactly one and reduces the number of edges by at least one. For a subset $F$ of edges in $G$, $G/ F$ denotes the graph obtained from $G$ by contracting all the edges (in some order) in $F$.  
We now formally define the notion of graph contractibility.

\begin{definition}\label{def:graph-contraction} $G$ is said to be contractible to $H$ if there is a surjective function $\psi: V(G) \rightarrow V(H)$ such that the following properties hold.
\begin{enumerate}
\item For each $h \in V(H)$, $\psi^{-1}(h)$, called the \emph{witness set} corresponding to $h$, is connected. 
\item For each $h, h’ \in V(H)$, $hh’ \in E(H)$ if and only if $E(\psi^{-1}(h), \psi^{-1}(h'))\neq \emptyset$.
\end{enumerate}
Then, we say that $G$ is \emph{contractible} to $H$ via the function $\psi$ and that $G$ has a $H$-\emph{witness structure} $\mathcal{W}=\{\psi^{-1}(h) \mid h \in V(H)\}$ which is the collection of all witness sets.
\end{definition}

In Definition \ref{def:graph-contraction}, a witness set that contains more than one vertex is called a \emph{big witness set} and the one that is a singleton set is called a \emph{small witness set} or \emph{singleton witness set}. Note that a witness structure $\mathcal{W}$ is a partition of $V(G)$. Also, if a vertex $v$ is in some big witness set $W$, then at least one neighbor of $v$ is also in $W$. Recall that the \textsc{$H$-Contractibility} problem takes as input a graph $G$ and decides whether $G$ is contractible to $H$ or not. Observe that this task is equivalent to determining if $G$ has a $H$-witness structure or not.

Now, we proceed to proving Theorems \ref{thm:c-5-contractiblity-np-hard} and \ref{thm:c-4-contractiblity-np-hard} in Sections \ref{sec:c5-contractiblity-bipartite} and \ref{sec:c4-contractiblity-bipartite}, respectively.

\section{C$_5$-Contractibility on Bipartite Graphs}
\label{sec:c5-contractiblity-bipartite}

In this section, we prove Theorem~\ref{thm:c-5-contractiblity-np-hard}. It is easy to verify that $C_5$-\textsc{Contractibility} is in \NP. Given an instance $\psi$ of \textsc{Positive NAE-SAT} with $N$ variables and $M$ clauses, we give a polynomial-time algorithm that outputs a bipartite graph $G$ equivalent to $\psi$. For the sake of simplicity, we describe the algorithm in two steps. In the first step, the algorithm constructs a non-bipartite graph $H$ equivalent to  $\psi$ (Lemmas \ref{lemma:Pos-NAE-SAT-C5-Contractibility-Part-1-forward} and \ref{lemma:Pos-NAE-SAT-C5-Contractibility-Part-1-backward}) and then in the second step, the algorithm constructs a bipartite graph $G$ that is equivalent to $H$ (Lemma \ref{lemma:Pos-NAE-SAT-C5-Contractibility-Part-2-forward-backward}). We remark that $G$ is obtained from $H$ by dividing some (and not all) of the edges of $H$. 

\subsection{Construction of $H$ and $G$}
Let $\{X_1, X_2, \dots, X_N\}$ and $\{C_1, C_2, \dots, C_M\}$ be the sets of variables and clauses, respectively, in $\psi$. The non-bipartite graph $H$ is constructed as follows. Refer to Figure~\ref{fig:c5-contractiblity-bipartite-np-hard} for an illustration. 
\begin{enumerate}
\item Add a set $V_\alpha=\{\alpha^0, \alpha^1, \alpha^2, \alpha^3, \alpha^4\}$ of five vertices that induce the 5-cycle $(\alpha^0, \alpha^1, \alpha^2, \alpha^3, \alpha^4)$. This set forms the ``base cycle'' in the witness structure. 
\item For every $i \in [N]$, add a set of five vertices that induce a 5-cycle $C^i=(x^0_i, x^1_i, x^2_i, x^3_i, x^4_i)$ and two sets of edges $\{x^0_i\alpha^0,$ $x^1_i\alpha^1, x^2_i\alpha^2, x^3_i\alpha^3, x^4_i\alpha^4\}$ and $\{x^0_i\alpha^1, x^1_i\alpha^2, x^2_i\alpha^3, x^3_i\alpha^4, x^4_i\alpha^0\}$. The variable gadget is designed so that there are two choices for $C^i$ to co-exist (in a $C_5$-witness structure) with the $C_5$ induced by $V_\alpha$. We will associate these two choices with a \true\ or \false\ assignment to the corresponding variable.
\item For every $j \in [M]$, add vertices $c_j$ and $b_j$ and a set $\{c_j\alpha^0, c_j\alpha^2, b_j\alpha^2, b_j\alpha^4\}$ of edges. The neighbours of $c_j$ and $b_j$ are defined so that $c_j$ will be in the same witness set as $\alpha^1$ (a non-neighbor of $c_j$) and $b_j$ will be in the same witness set as $\alpha^3$ (a non-neighbor of $b_j$).  
\item Finally, for every $i \in [N]$ and $j \in [M]$ such that $X_i$ appears in $C_j$, add edges $x^1_i c_j$ and $x^2_i b_j$. This step is the one that encodes the clause-variable relationship. Relevant variables are expected to help $c_j$ (and $b_j$) to be connected to witness sets containing $\alpha^1$ (and $\alpha^3$).   
\end{enumerate}
This completes the construction of $H$. 

For $p \in \{0, 1, 2, 3, 4\}$, define $X^p := \{x^p_i \mid i \in [N]\}$. Also, define $Y^c := \{c_j \mid j \in [M]\}$ and $Y^b := \{b_j \mid j \in [M]\}$. For an edge $uv \in E(H)$, let $\lambda(u, v)$ denote the new vertex added while subdividing $uv$ in the construction of $G$. Let $L=\{\alpha^0,\alpha^2,\alpha^4\} \cup X^1 \cup X^3$ and $R=\{\alpha^1,\alpha^3\} \cup X^0 \cup X^2 \cup X^4 \cup Y^c \cup Y^b$. Then, $\{L,R\}$ is a partition of $H$ into two parts where there are certain edges with both endpoints in the same part. We subdivide exactly these edges to obtain $G$.
\begin{enumerate}
\item \label{item:subdivide-x0-x4} Subdivide the edge $\alpha^0\alpha^4$.
\item For every $i \in [N]$, subdivide the edges $x^0_ix^4_i$, $x^0_i \alpha^1$, $x^1_i \alpha^2$, $x^2_i \alpha^3$, and $x^3_i \alpha^4$.
\item For every $i \in [N]$ and $j \in [M]$, subdivide the edge $x^2_i b_j$ if it exists.
\end{enumerate}
This completes the construction of $G$. 

We now argue that $G$ is a bipartite graph. Observe that $L$ and $R$ are independent sets in $G$. We will extend this partition $\{L,R\}$ of $H$ into a bipartition of $G$ as follows: $\lambda(\alpha^0, \alpha^4) \in R$ and for every $i \in [N]$,  $\lambda(x^0_i, x^4_i) \in L$, $\lambda(x^0_i,  \alpha^1) \in L$,  $\lambda(x^1_i,  \alpha^2) \in R$,  $\lambda(x^2_i, \alpha^3) \in L$ and $\lambda(x^3_i,  \alpha^4) \in R$. For every $i \in [N], j \in [M]$, if $x^2_i b_j \in E(H)$, then $\lambda(x^2_i, b_j) \in L$. See Figure~\ref{fig:c5-contractiblity-bipartite-np-hard} for an illustration. It is easy to verify that $\{L,R\}$ is a bipartition of $G$ and hence $G$ is a bipartite graph. 

We remark that the natural bipartite graph obtained from $H$ by subdividing all the edges may not be equivalent to $H$ in the context of \textsc{$C_5$-Contractiblity}. In Lemma~\ref{lemma:Pos-NAE-SAT-C5-Contractibility-Part-2-forward-backward}, we show that the set of edges of $H$ that are subdivided to obtain $G$ are  safe (in preserving contractiblity to $C_5$) to subdivide.

\begin{figure}[t]
  \begin{center}
    \includegraphics[scale=0.7]{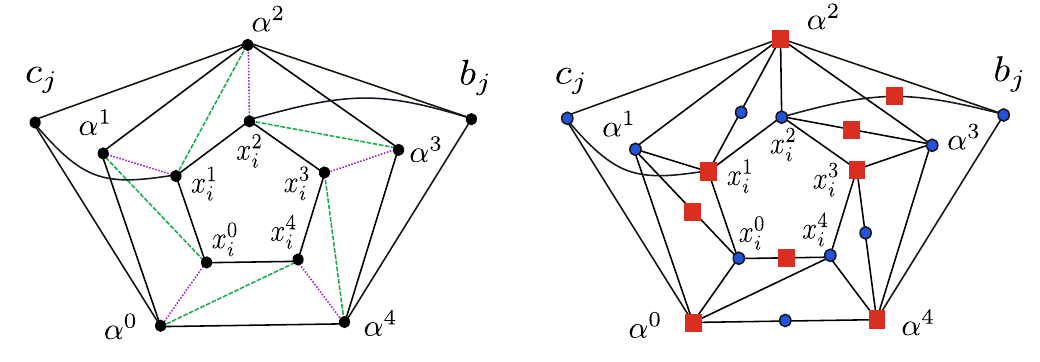}
    \end{center}
   \caption{(Left) The graph $H$ with certain edges highlighted as purple (dotted) edges denote setting variable $X_i$ to \true\ and as green (dashed) edges denote setting $X_i$ to \false, respectively.  (Right) The bipartite graph $G$ where blue (round) and red (squares) vertices denote a bipartition.\label{fig:c5-contractiblity-bipartite-np-hard}}
\end{figure}


\subsection{Equivalence of $H$ and $G$}

We show that $G$ and $H$ are equivalent in the context of \textsc{$C_5$-Contractiblity}. As $G$ is obtained from $H$ by subdividing some edges, one can obtain $H$ from $G$ by contracting some edges. Hence, if one can obtain a $C_5$ by contracting edges in $H$, then one can also obtain a $C_5$ by contracting edges in $G$ by first contracting $G$ to $H$ and then contracting $H$ to $C_5$. To prove the converse, we first argue that no vertex in $V(G) \setminus V(H)$ is a singleton witness set in any $C_5$-witness structure $\calW$ of $G$. Then, we show that deleting vertices of $V(G) \setminus V(H)$ from $\calW$ results in a $C_5$-witness structure $\calW'$ of $H$.

\begin{restatable}{lemma}{GequivH}
\label{lemma:Pos-NAE-SAT-C5-Contractibility-Part-2-forward-backward}
$H$ is a \yes-instance of \textsc{$C_5$-Contractibility} if and only if $G$ is a \yes-instance of \textsc{$C_5$-Contractibility}.
\end{restatable}
\begin{proof}
 As $G$ is obtained from $H$ by subdividing some edges, one can obtain $H$ from $G$ by contracting some edges. Hence, if one can obtain a $C_5$ by contracting edges in $H$, then one can also obtain a $C_5$ by contracting edges in $G$ by first contracting $G$ to $H$ and then contracting $H$ to $C_5$. Therefore, if $H$ is a \yes-instance of \textsc{$C_5$-Contractibility} then $G$ is a \yes-instance of \textsc{$C_5$-Contractibility}.

Suppose $G$ is a \yes-instance of \textsc{$C_5$-Contractibility} and $\calW=\{W^i \mid i \in [4] \cup \{0\}\}$ is a $C_5$-witness structure of $G$ where $E(W^i, W^j) \neq \emptyset$ if and only if $j=(i+1) \mod 5$. We first argue that no vertex in $V(G) \setminus V(H)$ is a singleton witness set in $\calW$. Suppose for some $uv \in E(H)$, $\lambda(u, v)$ is a singleton witness set. Without loss of generality, let $W^0 = \{\lambda(u, v)\}$. As $E(W^0, W^1)$ and $E(W^0, W^1)$ are non-empty sets in $G$ and $\lambda(u, v)$ is of degree two, either $u \in W^1$ and $v \in W^4$ or $v \in W^1$ and $u \in W^4$. In either case,  as $\calW$ is a $C_5$-witness structure of $G$, any path from $u$ to $v$ in $G - \{\lambda(u,  v)\}$ is of length at least three. We will now show that for any edge $uv$ in $H$ which is subdivided while constructing $G$, the length of a shortest path between $u$ and $v$ in $G - \{\lambda(u, v)\}$ is two. This will lead to a contradiction which will enable us to conclude that no vertex in $V(G) \setminus V(H)$ is a singleton witness set in $\calW$. 

Consider the following triples: $(\alpha^0, \alpha^4, x^4_i)$, $(x^0_i, x^4_i, \alpha^0)$, $(x^p_i, \alpha^{p + 1},  x^{p + 1}_i)$, and $(x^2_i, b_j, \alpha^2)$ for every $i \in [N]$, $j \in [M]$, and $p \in \{0, 1, 2, 3\}$. For any triple $(u', v', w')$, there is a path from $u'$ to $v'$ via $w'$ of length two in $G$ that does not contain $\lambda(u', v')$, i.e., $w'$ is a common neighbor of $u'$ and $v'$ in $G$. This implies that the vertex obtained by subdividing the edge $u'v'$ in $H$ while constructing $G$ cannot be a singleton witness set in $\calW$. However, these triples represent all edges of $H$ that were subdivided while constructing $G$. Hence,  there is no vertex in $V(G) \setminus V(H)$ which is a singleton witness set in $\calW$. Equivalently, for any $\lambda(u, v) \in V(G)$, if $\lambda(u, v)$ is contained in $W^{\star} \in \calW$, then $u \in W^{\star}$ or $v \in W^{\star}$. That is, for any edge $uv \in E(H)$ which was subdivided while constructing $G$, the vertices $u, v$ are in the same witness set or in adjacent witness sets. 

Let $\calW'$ be the partition of $V(H)$ obtained from $\calW$ by removing vertices in $V(G) \setminus V(H)$. Formally, $\calW' = \{W^{\star} \setminus (V(G) \setminus V(H)) \mid W^{\star} \in \calW\}$. Since, no vertex in $V(G)\setminus V(H)$ is a singleton witness set in $\calW$,  $\calW'$ contains five non-empty sets. Also, the endpoints of any edge $uv \in E(H) \setminus E(G)$ are either in the same witness set or in witness sets that are adjacent in $G$. In both the cases, each witness set of $\calW$ remains connected in $H$ even after deleting the vertices in $V(G) \setminus V(H)$. This implies that $\calW'$ is a $C_5$-witness structure of $H$.
\end{proof}

\subsection{Properties of a $C_5$-Witness Structure of $H$}

Before we state properties of $H$, we mention the following observation.

\begin{observation}
\label{obs:c5-bipart}
In any partition $\{X,Y\}$ of the vertices of an induced 5-cycle into 2 non-empty parts, $E(X,Y) \neq \emptyset$. 
\end{observation}

Note that this observation holds for any connected subgraph. Now, we state certain properties of vertex subsets in $H$ that we later use to show properties of a $C_5$-witness structure of $H$.

\begin{observation}
\label{obs:H-prop}
$X^0, X^1, X^2, X^3, X^4, Y^c$ and $Y^b$ are independent sets and $V_\alpha$ is a dominating set in $H$. Further, $X^0 \cup X^4 \cup Y^c \subseteq N(\alpha^0)$, $X^1 \cup X^0 \subseteq N(\alpha^1)$, $X^2 \cup X^1 \cup Y^c \cup Y^b \subseteq N(\alpha^2)$, $X^3 \cup X^2 \subseteq N(\alpha^3)$ and $X^4 \cup X^3 \cup Y^b \subseteq N(\alpha^4)$.
\end{observation}

Next, we show a property of a $C_5$-witness structure of $H$ that will be crucial to proving the correctness of the reduction. As we have indicated in the construction of $H$, we need a handle on the base cycle of the $C_5$-witness structure (for \yes-instances) which Lemma \ref{lemma:c5-ws} provides.

\begin{lemma}
\label{lemma:c5-ws}
In any $C_5$-witness structure of $G$, no two vertices of $V_\alpha=\{\alpha^0, \alpha^1, \alpha^2, \alpha^3, \alpha^4\}$ are in the same witness set.
\end{lemma}
\begin{proof} 
Suppose $\calW=\{W^i \mid i \in [4] \cup \{0\}\}$ is a $C_5$-witness structure of $H$ where $E(W^i, W^j) \neq \emptyset$ if and only if $j=(i\pm 1) \mod 5$. We argue that $V_\alpha$ has a non-empty intersection with each $W^i$. Suppose $V_\alpha \subseteq W^i$ for some $0 \leq i \leq 4$. Then, $W^{(i+2) \mod 5}=\emptyset$ and $W^{(i+3) \mod 5}=\emptyset$ leading to a contradiction. Suppose $V_\alpha$ intersects  exactly two witness sets. We will consider the cases when these sets are $W^0,W^1$ and $W^0,W^2$. The other cases are similar to these cases. If $V_\alpha$ intersects only with $W^0$ and $W^1$, then since $V_\alpha$ is a dominating set in $H$ it follows that $W^3=\emptyset$ and this leads to a contradiction. Suppose $V_\alpha$ intersects only with $W^0$ and $W^2$. From Observation \ref{obs:c5-bipart}, this implies that $E(W^0,W^2)\neq \emptyset$ leading to a contradiction. Suppose $V_\alpha$ intersects exactly four witness sets, say $W^0, W^1, W^2$, and $W^3$. Without loss of generality, assume $\alpha^0 \in W^0$. As $\alpha^1$ and $\alpha^4$ are adjacent to $\alpha^0$, we have $\{ \alpha^1, \alpha^4 \} \subseteq W^0 \cup W^1$. Then, one of $\alpha^2$ or $\alpha^3$ is in $W^2$ and the other is in $W^3$. However, as $\alpha^1 \alpha^2, \alpha^3 \alpha^4 \in E(H)$, neither $\alpha^2$ nor $\alpha^3$ can be in $W^3$ implying that $W^3=\emptyset$ and leading to a contradiction. 

Suppose $V_\alpha$ intersects exactly three witness sets. Without loss of generality, let $\alpha^0 \in W^0$. We consider the following cases.
\begin{itemize}
\item[-]Case (i) $V_\alpha$ intersects $W^0, W^1$ and $W^2$.
\item[-]Case (ii) $V_\alpha$ intersects $W^0, W^1$ and $W^4$.
\item[-]Case (iii) $V_\alpha$ intersects $W^0, W^2$ and $W^3$. This leads to contradiction as Observation \ref{obs:c5-bipart} implies $E(W^0,W^2 \cup W^3)\neq \emptyset$.
\item[-]Case (iv) $V_\alpha$ intersects $W^0, W^2$ and $W^4$. This leads to contradiction as Observation \ref{obs:c5-bipart} implies $E(W^2,W^0 \cup W^4)\neq \emptyset$.
\item[-]Case (v) $V_\alpha$ intersects $W^0, W^4$ and $W^3$. This is similar to Case (i). 
\item[-]Case (vi) $V_\alpha$ intersects $W^0, W^1$ and $W^3$. This is similar to Case (iv). 
\end{itemize}
\noindent Consider Case (i). As $\alpha^1$ and $\alpha^4$ are adjacent to $\alpha^0$, we have $\{ \alpha^1, \alpha^4 \} \subseteq W^0 \cup W^1$. Then, at least one of $\alpha^2$ or $\alpha^3$ is in $W^2$ and since $\alpha^2 \alpha^3 \in E(H)$, neither $\alpha^2$ nor $\alpha^3$ can be in $W^0$. Thus, we have $\{\alpha^2, \alpha^3\} \subseteq  W^1 \cup W^2$. Since $E(W^0, W^3)=\emptyset$ and $E(W^1, W^3)=\emptyset$, we have $W^3 \cap N(\alpha^0) = \emptyset$, $W^3 \cap N(\alpha^1)  = \emptyset$ and $W^3 \cap N(\alpha^4) = \emptyset$. From Observation \ref{obs:H-prop}, this implies that $W^3 \subseteq X^2$. Similarly, since $E(W^1, W^4)=\emptyset$ and $E(W^2, W^4)=\emptyset$, we have $W^4 \cap N(\alpha^2) = \emptyset$ and $W^4 \cap N(\alpha^3) = \emptyset$. From Observation \ref{obs:H-prop}, this implies $W^4 \subseteq (X^0 \cup X^4)$. However, by the construction, $E(X^2, X^0 \cup X^4)=\emptyset$ implying that $E(W^3, W^4)=\emptyset$ which leads to a contradiction. 

Let us now consider Case (ii). Recall that $\alpha^0 \in W^0$. Then, either $\alpha^1 \in W^0 \cup W^1$ or $\alpha^1 \in W^0 \cup W^4$. As both these cases are similar, we consider the case when $\alpha^1 \in W^0 \cup W^1$. Suppose $\alpha^1 \in W^1$.   
Then, we have $\{\alpha^1, \alpha^2\} \subseteq W^0 \cup W^1$ since $\alpha^1 \alpha^2 \in E(H)$. We will show that this leads to a contradiction.
At least one of $\alpha^3$ or $\alpha^4$ is in $W^4$ and since $\alpha^3 \alpha^4 \in E(H)$, neither $\alpha^3$ nor $\alpha^4$ can be in $W^1$. Thus, we have $\{\alpha^3, \alpha^4\} \subseteq W^0 \cup W^4$. Since $E(W^0, W^3)=\emptyset$ and $E(W^1, W^3)=\emptyset$, we have $W^3 \cap N(\alpha^0) = \emptyset$, $W^3 \cap N(\alpha^1) = \emptyset$ and $W^3 \cap N(\alpha^2) = \emptyset$. From Observation \ref{obs:H-prop}, this implies $W^3 \subseteq X^3$. Similarly, since $E(W^0, W^2)=\emptyset$ and $E(W^4, W^2)=\emptyset$, we have $W^2 \cap N(\alpha^0)= \emptyset$, $W^2 \cap N(\alpha^3) = \emptyset$ and $W^2 \cap N(\alpha^4)= \emptyset$. From Observation \ref{obs:H-prop}, this implies $W^2 \subseteq X^1$. However, by construction, $E(X^1, X^3)=\emptyset$  implying that $E(W^2, W^3)=\emptyset$ which leads to a contradiction. 

Suppose $\alpha^1 \in W^0$. If $\alpha^2 \in W^0$, then one of $\alpha^3$ or $\alpha^4$ is in $W^1$ and the other is in $W^4$ resulting in an edge between $W^1$ and $W^4$. Thus, $\alpha^2 \in W^1$ or $\alpha^2 \in W^4$. As these cases are similar, we only consider $\alpha^2 \in W^1$. Then we once again have $\{\alpha^1, \alpha^2\} \subseteq W^0 \cup W^1$ which leads to a contradiction. 
\end{proof}

\subsection{Equivalence of $H$ and $\psi$}
Now, we are ready to establish the equivalence of $\psi$ and $H$.

\begin{restatable}{lemma}{satfive}
\label{lemma:Pos-NAE-SAT-C5-Contractibility-Part-1-forward}
If $\psi$ is a \yes-instance of \textsc{Positive NAE-SAT} then $H$ is a \yes-instance of \linebreak \textsc{$C_5$-Contractibility}.
\end{restatable}
\begin{proof}
Suppose $\pi : \{X_1, X_2, \dots, X_N\} \mapsto \{\true, \false\}$ is a not-all-equal satisfying assignment of $\psi$. Define the following partition of $V(H)$.
\begin{align*}
W^0  :=  &\ \{\alpha^0\} \cup \{x^0_i \mid i \in [N], \pi(X_i) = \true\} \cup \{x^4_i \mid i \in [N], \pi(X_i) = \false\}, \\
W^1  := &\ \{\alpha^1\} \cup \{x^1_i \mid i \in [N], \pi(X_i) = \true\} \cup \{x^0_i \mid i \in [N], \pi(X_i) = \false\}\\ 
&\ \cup \{c_j \mid j \in [M]\}, \\
W^2 :=  &\ \{\alpha^2\} \cup \{x^2_i \mid i \in [N], \pi(X_i) = \true\} \cup \{x^1_i \mid i \in [N], \pi(X_i) = \false\}, \\
W^3 := &\ \{\alpha^3\} \cup \{x^3_i \mid i \in [N], \pi(X_i) = \true\} \cup \{x^2_i \mid i \in [N], \pi(X_i) = \false\}\\
&\ \cup \{b_j \mid j \in [M]\}, \\
W^4 := &\ \{\alpha^4\} \cup \{x^4_i \mid i \in [N], \pi(X_i) = \true\} \cup \{x^3_i \mid i \in [N], \pi(X_i) = \false\}, 
\end{align*}
Clearly $W^0$, $W^2$, and $W^4$ are connected sets. For any $j \in [M]$, there exists $i \in [N]$ such that $x^1_i \in W^1$ (since $\pi$ sets at least one of the variables in $C_j$ to \true) and $i' \in [N]$ such that $x^2_{i'} \in W^3$ (since $\pi$ sets at least one of the variables in $C_j$ to \false). Also, $c_jx^1_i, b_jx^2_{i'} \in E(H)$. As for every $i \in [N]$, $\alpha^1$ is adjacent to $x^1_i$ and $\alpha^3$ is adjacent to $x^2_i$, it follows that $W^1$ and $W^3$ are connected sets. Now, it is easy to verify that $\{W^0, W^1, W^2, W^3, W^4\}$ is a $C_5$-witness structure. 
\end{proof}

In the proof of the converse of Lemma \ref{lemma:Pos-NAE-SAT-C5-Contractibility-Part-1-forward}, we crucially use Lemma \ref{lemma:c5-ws}. That is, if $H$ is contractible to a 5-cycle, then in any $C_5$-witness structure $\{W^0, W^1, W^2, W^3, W^4\}$ with $E(W^i, W^j) \neq \emptyset$ if and only if $j=(i \pm 1) \mod 5$, each of the five witness sets has a non-empty intersection with $V_\alpha$. This structure along with a couple of other properties translates to a not-all-equal satisfying assignment of $\psi$. 

\begin{lemma}
\label{lemma:Pos-NAE-SAT-C5-Contractibility-Part-1-backward}
If $H$ is a \yes-instance of \textsc{$C_5$-Contractibility} then $\psi$ is a \yes-instance of \textsc{Positive NAE-SAT}.
\end{lemma}
\begin{proof}
Suppose $\calW = \{W^0, W^1, W^2, W^3, W^4\}$ is a $C_5$-witness structure of $H$ where $E(W^i, W^j) \neq \emptyset$ if and only if $j=(i\pm 1) \mod 5$. Then, by Lemma \ref{lemma:c5-ws}, $V_\alpha$ has a non-empty intersection with each $W^i$. Without loss of generality, let $\alpha^p \in W^p$ for every $p \in \{0, 1, 2, 3, 4\}$. We first argue that for any $i \in [N]$, the set $S_i=\{x^0_i, x^1_i, x^2_i, x^3_i, x^4_i\}$ also has a non-empty intersection with each $W^j$. Suppose $S_i \cap W^0 =\emptyset$. Then, as $\alpha^0$ is adjacent to $x^0_i, x^4_i$ and $x^0_ix^4_i \in E(H)$, either $\{x^0_i, x^4_i\} \subseteq W^1$ or $\{x^0_i, x^4_i\} \subseteq W^4$. As $\alpha^4x^4_i, \alpha^1x^0_i   \in E(H)$,  both these cases contradict the fact that $E(W^1, W^4) = \emptyset$. Using the similar arguments, it follows that $S_i$ has a non-empty intersection with each $W^j$. 

Next, we claim that for each $i \in [N]$ and $0 \leq p \leq 4$, $x^p_i \in W^p \cup W^{(p+1) \mod 5}$. This is due to the fact that $x^p_i$ is adjacent with $\alpha^p$ and $\alpha^{p + 1\ (\text{mod}\ 5)}$. Now, we show that for each $i \in [N]$ and $0 \leq p \leq 4$, $x^p_i \in W^p$ if and only if $x^{(p+1) \mod 5}_i \in W^{(p+1) \mod 5}$ and $x^p_i \in W^{(p+1) \mod 5}$ if and only if $x^{(p+1) \mod 5}_i \in W^{(p+2) \mod 5}$. If $x^0_i \in W^0$ and $x^1_i \notin W^1$,  then $E(W^0, W^2) \cup E(W^0, W^3) \cup E(W^2, W^4) \neq \emptyset$ leading to a contradiction. If $x^0_i \in W^1$ and $x^1_i \notin W^2$,  then $E(W^1, W^3) \cup E(W^0, W^2) \cup E(W^1, W^4) \neq \emptyset$ leading to a contradiction. Similar arguments hold for $x^1_i, x^2_i,x^3_i$ and $x^4_i$. This is indicated by the collections of purple (dotted) edges and green (dashed) edges in Figure \ref{fig:c5-contractiblity-bipartite-np-hard}. We will associate these two choices with setting $X_i$ to \true\ and to \false, respectively.

We now construct an assignment $\pi: \{X_1, X_2, \dots, X_N\} \mapsto \{\true, \false\}$. Consider the witness set $W^1$. For each $i \in [N]$, if $x^1_i \in W^1$ then set $\pi(X_i) = \true$, otherwise ($x^1_i \in W^2$) set $\pi(X_i) = \false$. We argue that $\pi$ is a not-all-equal satisfying assignment for $\psi$. We show that for each $j \in [M]$, $c_j \in W^1$ and $b_j \in W^3$, further, the clause $C_j$ has variables $X_i$ and $X_{i'}$ such that $x^1_i \in W^1$ and $x^2_{i'} \in W^3$. Observe that $c_j$ (being adjacent with $\alpha^0$ and $\alpha^2$) is in the same witness set that has $\alpha^1$ and $b_j$ (being adjacent with $\alpha^2$ and $\alpha^4$) is in the same witness set that has $\alpha^3$. Thus, for each $j \in [M]$, $c_j \in W^1$ and $b_j \in W^3$. By the property of witness structures, $W^1$ and $W^3$ are connected sets. As the only vertices outside $V_\alpha$ that are adjacent to $c_j$ are vertices $x^1_i$ corresponding to variables $X_i$ appearing in $C_j$, it follows that $C_j$ has a variable $X_i$ such that $x^1_i \in W^1$. Similarly, as the only vertices outside $V_\alpha$ that are adjacent to $b_j$ are vertices $x^2_i$ corresponding to variables $X_i$ appearing in $C_j$, it follows that $C_j$ has a variable $X_{i'}$ such that $x^2_{i'} \in W^3$. 
\end{proof}

\section{$C_4$-Contractiblity on Biparitite Graphs}
\label{sec:c4-contractiblity-bipartite}
In this section, we prove Theorem~\ref{thm:c-4-contractiblity-np-hard}. It is easy to verify that $C_4$-\textsc{Contractibility} is in \NP. Given an instance $\psi$ of \textsc{Positive NAE-SAT} with $N$ variables and $M$ clauses, we give a polynomial-time algorithm that outputs a bipartite graph $G$ equivalent to $\psi$ (Lemmas \ref{lemma:Pos-NAE-SAT-Constrained-C4-Contraction-forward} and \ref{lemma:Pos-NAE-SAT-Constrained-C4-Contraction-backward}).

\subsection{Construction of $G$}
Let $\{X_1, X_2, \dots, X_N\}$ and $\{C_1, C_2, \dots, C_M\}$ be the sets of variables and clauses, respectively, in $\psi$. The graph $G$ with a partition $\{V, V'\}$ of its vertex set is constructed as follows. See Figure~\ref{fig:c4-contractiblity-bipartite-np-hard} for an illustration.

\begin{enumerate}
\item Add vertices $t, f$ to $V$, vertices $t', f'$ to $V'$ and edges $tt', ff'$ to $E(G)$. This set would eventually form the ``base cycle'' in the witness structure. 
\item For every $i \in [N]$, add vertices $x_i, y_i, z_i$ to $V$ and $x'_i, y'_i, z'_i$ to $V'$ corresponding to the variable $X_i$. Further, make every vertex in $\{x'_i, y'_i, z'_i\}$ adjacent to every vertex in $\{x_i, t, f\}$ and every vertex in $\{x_i, y_i, z_i\}$ adjacent to every vertex in $\{x'_i, t', f'\}$. Let $X=\{x_i \mid i \in [N]\}$, $X'=\{x'_i \mid i \in [N]\}$, $Y=\{y_i \mid i \in [N]\}$, $Y'=\{y'_i \mid i \in [N]\}$, $Z=\{z_i \mid i \in [N]\}$, $Z'=\{z'_i \mid i \in [N]\}$. The neighborhood of $X'$ is set so that every element of $X'$ is in the witness set containing $t$ or $f$. This forces every element of $X$ to be respectively in the witness set containing $t'$ or $f'$. These binary choices would be associated with setting the corresponding variable to \true\ or \false. The sets $Y$, $Y'$, $Z$, $Z'$ are added for technical reasons.
\item  For every $j \in [M]$, add vertices $c_j, b_j$ to $V$, $c'_j, b'_j$ to $V'$ and edges $c_jf', b_jf'$, $c'_jt$, $b'_jt$ to $E(G)$ corresponding to clause $C_j$. Let $C=\{c_j \mid j \in [M]\}$, $C'=\{c'_j \mid j \in [M]\}$, $B=\{b_j \mid j \in [M]\}$, $B'=\{b'_j \mid j \in [M]\}$. Subsequently, we will add more vertices (sets $D$ and $D'$ defined subsequently) adjacent to vertices in $C \cup B \cup C' \cup B'$ so that no vertex in $B \cup C$ is in a witness set that is non-adjacent to the one containing $t$ and no vertex in $B' \cup C'$ is in a witness set that is non-adjacent to the one containing $f'$.
\item  For every $i \in [N]$ and $j \in [M]$, if $X_i$ appears in $C_j$ then add edges $c_jx'_i, b_jx'_i$, $x_ic'_j$, and $x_ib'_j$ to $E(G)$. This step is the one that encodes the clause-variable relationship. Relevant variables are expected to help clause vertices to be connected to witness sets containing them.   
\item Let $\calD$ denote the following collection of pairs of vertices: 
$\{\{t, f\}, \{t', f'\}\}$  $\bigcup \{\{t, c_j\}, \{t, b_j\},$  $\{f', c'_j\}, \{f', b'_j\}  \mid j \in [M] \}$. Note that for any pair of vertices in $\calD$, either both elements of the pair are in $V$ or both are in $V'$. For every pair $\{u, v\}$ of vertices in $\calD$ that are in $V$, add three vertices $d'_{u,v,1}$, $d'_{u, v, 2}$, $d'_{u, v, 3}$ to $V'$ and make them adjacent to both $u, v$. For every pair $\{u, v\}$ of vertices in $\calD$ that are in $V'$, add three vertices $d_{u,v,1}$, $d_{u, v, 2}$, $d_{u, v, 3}$ to $V$ and make them adjacent to both $u, v$. The pairs in $\calD$ are the ones that should not be in non-adjacent witness sets and the common neighbors are added to achieve this property. 
\end{enumerate}
This completes the construction of $G$. 

\begin{figure}[t]
  \begin{center}
    \includegraphics[scale=0.7]{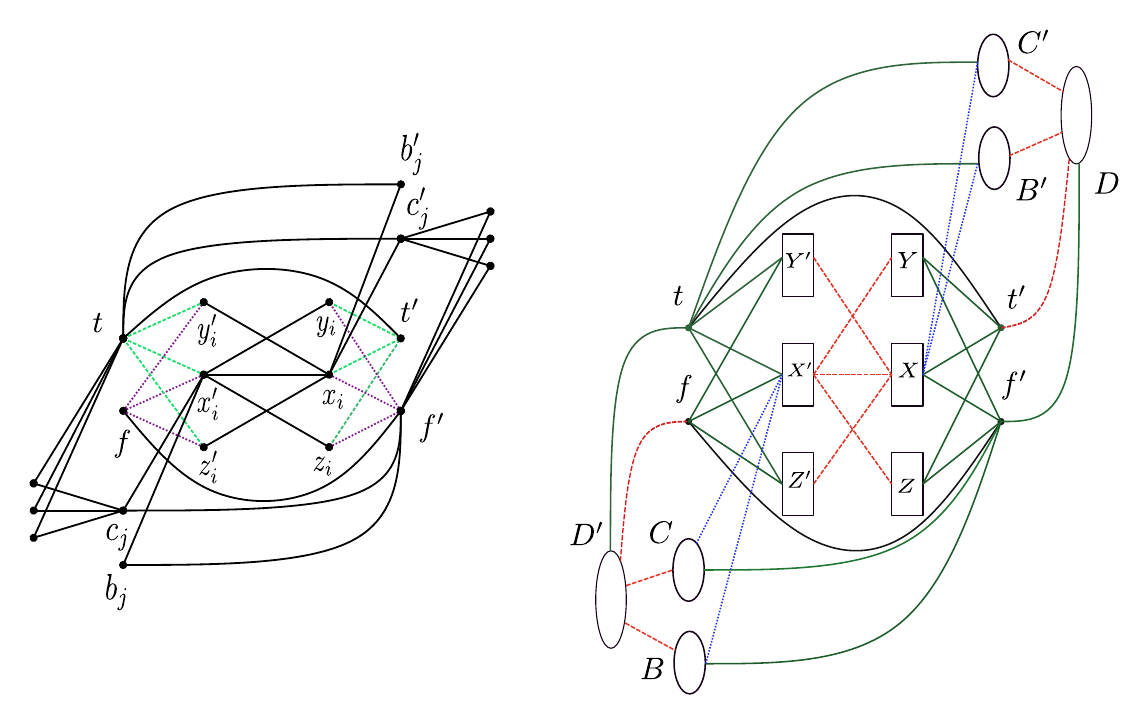}
    \end{center}
   \caption{(Left) The graph $G$ where only three vertices each in $D$ and $D'$ shown with purple (dotted) edges denote setting variable $X_i$ to \true\ and green (dashed) edges denote setting $X_i$ to \false. (Right) Adjacency relation between different subsets of vertices.  \label{fig:c4-contractiblity-bipartite-np-hard}}
\end{figure}

As the reduction always adds edges with one of its endpoints in $V$ and the other endpoint in $V'$, $G$ is a bipartite graph with bipartition $\{ V, V' \}$. Let $D=\{d_{u, v, p} \mid \{u, v\} \in \calD, u, v \in V\ \text{and}\ p \in [3]\}$ and $D'=\{d'_{u, v, p} \mid \{u, v\} \in \calD, u, v \in V'\ \text{and}\ p \in [3]\}$.

\subsection{Properties of a Nice $C_4$-Witness Structure of $G$}

Now, we show that if $G$ is contractible to a 4-cycle, then there is a $C_4$-witness structure of $G$ satisfying certain nice properties. For this purpose, we introduce the following notion of a {\em nice $C_4$-witness structure}. 

\begin{definition}
A $C_4$-witness structure of $G$ is a nice $C_4$-witness structure if the following properties hold.\\
(P1) For every pair $\{u,v\}$ in $\mathcal{D}$, $u$ and $v$ are in the same or adjacent witness sets.\\
(P2) Every vertex in $D \cup D'$ is in a big witness set. Further, every vertex in $D'$ is in the same witness set as $t$ and every vertex in $D$ is in the same witness set as $f'$.
\end{definition}

\noindent Next, we show the existence of a nice $C_4$-witness structure for \yes-instances.

\begin{restatable}{lemma}{niceWS}
\label{lemma:nice-c4-ws}
If $G$ is contractible to a 4-cycle, then there is a nice $C_4$-witness structure of $G$.
\end{restatable}
\begin{proof}
Let $\calW=\{ W^0, W^1, W^2, W^3\}$ be a $C_4$-witness structure of $G$ such that $E(W^i, W^j) \neq \emptyset$ if and only if $j=(i+1) \mod 4$. We first show that $\calW$ satisfies Property (P1). Assume for the sake of contradiction that there is a pair $\{u, v\} \in \calD$ such that $u \in W^i$ and $v \in W^{(i+2) \mod 4}$. Recall there are three pairwise non-adjacent vertices $d_{u, v, 1}$, $d_{u, v, 2}$ and $d_{u, v, 3}$ whose neighborhood in $G$ is $\{u, v\}$. It follows that each of these vertices is in $W^{(i+1) \mod 4}$ or in $W^{(i+3) \mod 4}$. Then, one of these sets, say $W^j$, contains $d_{u, v, 1}$ and $d_{u, v, 2}$. However, $W^j$ has no vertex that is adjacent to either of these vertices contradicting the fact that $W^j$ is a connected set. 

Next, we show that every vertex in $D \cup D'$ is in a big witness set. Assume for the sake of contradiction that there is a vertex, say $d_{u, v, 1} \in D \cup D'$, which is in a singleton witness set $W^i$. As $W^i$ is adjacent to $W^{(i+1) \mod 4}$ and $W^{(i+3) \mod 4}$ while $d_{u, v, 1}$ is adjacent to only $u$ and $v$, it follows that $u \in W^{(i+3) \mod 4}$ and $v \in W^{(i+1) \mod 4}$ Then, we get $d_{u, v, 2}, d_{u, v, 3} \in W^{(i+2) \mod 4}$. However, $N(d_{u, v, 2}) =  N(d_{u, v, 3}) = \{u, v\} \subseteq W^{(i+1) \mod 4} \cup W^{(i+3) \mod 4}$. This contradicts the fact that $W^{(i+2) \mod 4}$ is a connected set.

Subsequently, we assume that $\calW$ satisfies Property (P1) and every vertex in $D \cup D'$ is in a big witness set. Without loss of generality assume $t \in W^0$. Let $\mu(\calW)$ be the number of vertices in $D'$ that are not in $W^0$. If $\mu(\calW) = 0$ then it follows that every vertex in $D'$ is in the same witness set as $t$. Suppose $\mu(\calW) \ge 1$. Let $d'_{t, u, p}$ be a vertex in $D' \setminus W^0$ for some $u \in \{f\} \cup C \cup B$ and $p \in [3]$. As $t$ is adjacent to $d'_{t, u, p}$ and $t \in W^0$, $d'_{t, u, p}$ is either in $W^1$ or in $W^3$. Without loss of generality, suppose $d'_{t, u, p} \in W^1$. As every vertex in $D \cup D'$ is in a big witness set, $W^1$ is a big witness set. As $d'_{t, u, p}$ is adjacent to only $t$ and $u$, we have $u \in W^1$. Define $W^{0}_{\star} = W^0 \cup \{d'_{t, u, p}\}$ and $W^{1}_{\star} = W^1 \setminus \{d'_{t, u, p}\}$.
It is easy to verify that $\calW_{\star} = \{W^{0}_{\star}, W^{1}_{\star}, W^2, W^3\}$ is a $C_4$-witness structure of $G$. Moreover, $\mu(\calW_{\star}) < \mu(\calW)$. Hence, by repeating the this process at most $|D'|$ times, we obtain a $C_4$-witness structure of $G$ in which every vertex in $D'$ is in the same witness set as $t$. Using identical arguments, we can obtain a witness structure which also satisfies the property that every vertex in $D$ is in the same witness set as $f'$. 
\end{proof}

Now, we show a property of a nice $C_4$-witness structure of $G$ that will be crucial to proving the correctness of the reduction.

\begin{restatable}{lemma}{nicerWS}
\label{lemma:nicer-c4-ws}
In any nice $C_4$-witness structure of $G$, every pair of vertices in $\{t,t',f,f'\}$ are in different witness sets.
\end{restatable}
\begin{proof}
We first show that $t$, $f'$ are in different witness sets. Assume, for the sake of contradiction, that $t$ and $f'$ are in the same witness set, say $W^0$. By construction, every vertex in $V(G) \setminus \{t, f'\}$ is adjacent to either $t$ or $f'$. That is, $\{t,f'\}$ is a dominating set in $G$ implying that $W^2$ is empty since every (non-empty) witness set containing any vertex in $V(G) \setminus \{t,f'\}$ is adjacent to $W^0$. Hence, we conclude that $t$ and $f'$ are in different witness sets.  

Next, we prove that $t', f$ are in different witness sets. Assume, for the sake of contradiction, that $t', f$ are in the same witness set, say $W^0$. As $W^2$ is not adjacent to $W^0$, this implies that $W^2 \subseteq C \cup B \cup C' \cup B' \cup D \cup D'$. By Property (P2) of a nice $C_4$-witness structure of $G$, we have $W^2 \subseteq C \cup B \cup C' \cup B'$. As $C \cup B \cup C' \cup B'$ is an independent set and $W^2$ is a connected set, it follows that $W^2$ is a singleton set. Suppose $W^2 = \{c_j\}$ for some $j \in [M]$. By construction, $N(c_j) \subseteq \{f'\} \cup X' \cup D'$. As every vertex in $X' \cup \{f'\}$ is adjacent to $f \in W^0$, we have $X' \cup \{f'\} \subseteq W^1 \cup W^3$. Once again by Property (P2) of a nice $C_4$-witness structure of $G$, all the vertices in $D$ are in the same witness set that contains $f'$. As $W^1$ and $W^3$ are adjacent to $W^2$, we have $(D' \cup X' \cup \{f'\}) \cap W^1 \neq \emptyset$ and $(D' \cup X' \cup \{f'\}) \cap W^3 \neq \emptyset$. By construction, $N(c_j) \cap (X' \cup \{f'\}) = N(b_j) \cap (X' \cup \{f'\})$. This implies that $b_j$ is in $W^0$ as $W^2 = \{c_j\}$. By our assumption, $W^0 \setminus \{b_j\}$ contains $t', f$ and hence is non-empty. However, as $N(b_j) \subseteq W^1 \cup W^3$, $b_j$ is not adjacent with any other vertex in $W^0$. This contradicts the fact that $W^0$ is a connected set. A similar argument holds if $W^2 = \{c'_j\}$ or $W^2 = \{b_j\}$ or $W^2 = \{b'_j\}$. Thus, $t', f$ are in different witness sets. 

Now, we show that $t, t'$ are in different witness sets and $f, f'$ are in different witness sets. Assume, for the sake of contradiction, that $t$ and $t'$ are in the same witness set, say $W^0$. As $W^2$ is not adjacent with $W^0$, this implies that $W^2 \subseteq \{f, f'\} \cup C \cup B \cup D$. Recall that the pairs $\{t, f\}$, $\{t', f'\}$, $\{t, c_j\}$ and $\{t, b_j\}$ are in $\calD$ for every $j \in [M]$. As $t, t' \in W^0$, by Property (P1) of a nice $C_4$-witness structure of $G$, we have that $f, f', c_j, b_j \not\in W^2$ for any $j \in [M]$. By Property (P2), we have $(D \cup D') \cap W^2 = \emptyset$. This contradicts the fact that $W^2$ is a nonempty set. Hence, our assumption is wrong, and $t, t'$ are indeed in different set. Using a similar argument, it is easy to see that $f, f'$ are in different witness sets as well. 

Finally, we show that $t, f$ are in different witness sets and $t', f'$ are in different witness sets. Assume, for the sake of contradiction, that $t, f$ are in the same witness set, say $W^0$. As $W^2$ is not adjacent to $W^0$, we have $W^2 \subseteq X \cup Y \cup Z \cup C \cup B \cup D$.
As $f'$ is adjacent to $f$, we have $f' \in W^1 \cup W^3$, As $\{t, f'\} \cap W^2 = \emptyset$, by Property (P2) of a nice $C_4$-witness structure of $G$, we have $D \cap W^2 = \emptyset$. Hence, $W^2 \subseteq X \cup Y \cup Z \cup C \cup B$. Recall that for any $j \in [M]$, the pair $\{t, c_j\}$ is in $\calD$. As $t \in W^0$, by Property (P1) of a nice $C_4$-witness structure of $G$, we get that $c_j \not\in W^2$ for any $j \in [M]$. Using symmetric arguments, we can conclude $b_j \not\in W^2$ for any $j \in [M]$ as well. This implies that $W^2 \subseteq X \cup Y \cup Z$. By construction, $X \cup Y \cup Z$ is an independent set in $G$. As $W^2$ is a connected set, it follows that it is a singleton witness set. Suppose $W^2 = \{x_i\}$ for some $i \in [N]$. Recall that $t', f'$ are adjacent to $t$ and $f$, respectively, and are adjacent to $x_i$. As $t, f \in W^0$ and $x_i \in W^2$, we have $t', f' \in W^1 \cup W^3$. However, the pair $\{t', f'\}$ is in $\calD$. By Property (P1) of a nice $C_4$-witness structure of $G$, either $\{t', f'\} \subseteq W^1$ or $\{t', f'\} \subseteq W^3$. Without loss of generality, suppose $\{t', f'\}  \subseteq W^3$. Recall that $N(y'_i) = N(z'_i) = \{t, f, x_i\}$. Then, $t, f \in W^0$ and $\{x_i\} = W^2$ imply that $y'_i, z'_i \in W^1 \cup W^3$. As $N(y'_i) \subseteq W^0 \cup W^2$, if $y'_i \in W^1$, then $W^1 = \{y'_i\}$. A similar statement holds for $z'_i$. This implies that either $y'_i$ or $z'_i$ is present in $W^3$. Suppose $z'_i \in W^3$.
Then, $W^3 \setminus \{z'_i\}$ contains $t', f'$, and hence is non-empty. As $N(z'_i) \subseteq W^0 \cup W^2$, $z'_i$ is not adjacent with any other vertex in $W^3$. This contradicts the fact that $W^3$ is connected. Hence, our assumption is wrong and $t, f$ are in different witness sets. A similar argument holds if $W^2 = \{y_i\}$ or $W^2 = \{z_i\}$. Using a symmetric argument, it follows that $t', f'$ are in different witness sets. 
\end{proof}

\subsection{Equivalence of $G$ and $\psi$}

Now, we are ready to establish the equivalence of $\psi$ and $G$.

\begin{restatable}{lemma}{satfour}
\label{lemma:Pos-NAE-SAT-Constrained-C4-Contraction-forward}
If $\psi$ is a \yes-instance of \textsc{Positive NAE-SAT} then $G$ is a \yes-instance of \linebreak \textsc{$C_4$-Contractibility}.
\end{restatable}
\begin{proof}
Suppose $\pi : \{X_1, X_2, \dots, X_N\} \mapsto \{\true, \false\}$ is a not-all-equal satisfying assignment of $\psi$. Define the following partition of $V(G)$.
\begin{align*}
W^0 &\ := \{t\} \cup \{x'_i, y'_i, z'_i \mid i \in [N]\ \text{and}\ \pi(X_i) = \true\}\ \cup D', \\
W^1 &\ := \{t'\} \cup \{x_i, y_i, z_i \mid i \in [N]\ \text{and}\ \pi(X_i) = \true\} \cup B' \cup C', \\
W^2 &\ := \{f'\} \cup \{x_i, y_i, z_i \mid i \in [N]\ \text{and}\ \pi(X_i) = \false\} \cup D,  \ \text{and}\\
W^3 &\ := \{f\} \cup \{x'_i, y'_i, z'_i \mid i \in [N]\ \text{and}\ \pi(X_i) = \false\} \cup B \cup C.
\end{align*}
As $t$ is adjacent to every vertex in $X' \cup Y' \cup Z' \cup D'$, and $f'$ is adjacent to every vertex in $X \cup Y \cup Z \cup D$, $W^0$ and $W^2$ are connected sets in $G$. Further, by construction, $E(W^0,W^2)= \emptyset$ and $E(W^1,W^3)= \emptyset$. $W^1$ is a connected set since $X \cup Y \cup Z \subseteq N(t')$ and for each $j \in [M]$, there exists $i \in [N]$ such that $x_i \in W^1$ (corresponding to a variable in $C_j$ set to \true) and $c'_jx_i, b'_jx_i \in E(G)$. Similarly, $W^3$ is also a connected set. The edges $tt'$ and $ff'$, respectively, ensure that $W^0$ is adjacent to $W^1$ and $W^3$ is adjacent to $W^2$. As for any $i \in [N]$, $x'_i$ is adjacent with $t$ and $f$ and $x'_i \in W^0 \cup W^3$, it follows that $W^0$ and $W^3$ are adjacent. Similarly, $W^1$ and $W^2$ are adjacent. Hence, $\{W^0, W^1, W^2, W^3\}$ is a $C_4$-witness structure.
\end{proof}

Now, we proceed to show the converse of Lemma \ref{lemma:Pos-NAE-SAT-Constrained-C4-Contraction-forward}. We crucially use the properties of a nice $C_4$-witness structure. This structure along with certain other properties help to obtain a not-all-equal satisfying assignment of $\psi$.

\begin{lemma}
\label{lemma:Pos-NAE-SAT-Constrained-C4-Contraction-backward}
If $G$ is a \yes-instance of \textsc{$C_4$-Contractibility} then $\psi$ is a \yes-instance of \textsc{Positive NAE-SAT}.
\end{lemma}
\begin{proof}
Suppose $\calW = \{W^0, W^1, W^2, W^3\}$ is a $C_4$-witness structure of $G$ where $E(W^i, W^j) \neq \emptyset$ if and only if $j=(i\pm 1) \mod 4$. From Lemmas \ref{lemma:nice-c4-ws} and \ref{lemma:nicer-c4-ws}, we may assume that $\calW$ is a nice $C_4$-witness structure in which every pair of vertices in $\{t,t',f,f'\}$ are in different witness sets. As $\{t, f\}$ and $\{t', f'\}$ are in $\calD$, by Property (P1) of a nice $C_4$-witness structure of $G$, $t$ and $f$ are in adjacent witness sets and $t'$ and $f'$ are in adjacent witness sets. Hence, without loss of generality, we may assume that $t \in W^0$, $t' \in W^1$, $f' \in W^2$, and $f \in W^3$. Also, by Property (P2) of a nice $C_4$-witness structure of $G$, we have $D' \subseteq W^0$ and $D \subseteq W^2$. 

For each $i \in [N]$, $x'_i$ is adjacent to $t, f$ and $x_i$ is adjacent to $t', f'$. Therefore, $x_i \notin W^0 \cup W^3$, $x'_i \notin W^1 \cup W^2$ and we have $X' \subseteq W^0 \cup W^3$ and $X \subseteq W^1 \cup W^2$.
Further, since $x_i x'_i \in E(G)$, it follows that $x_i \in W^1$ if and only if $x'_i \in W^0$ and $x_i \in W^2$ if and only if $x'_i \in W^3$. Refer to Figure~\ref{fig:c4-contractiblity-bipartite-np-hard} for an illustration where these two choices are indicated by the purple (dotted) edges and green (dashed) edges. We will associate these two choices with setting the variable $X_i$ to $\true$ or $\false$, respectively. Consider a vertex $c_j \in C$ for some $j \in [M]$. As $f' \in W^2$ and $f'c_j \in E(G)$, it follows that $c_j\notin W^0$. Also, since $t \in W^0$ and $\{t, c_j\}$ is in $\calD$, by Property (P1) of a nice $C_4$-witness structure of $G$, it follows that $c_j$ is not in $W^2$. As $N(c_j) \subseteq W^0 \cup W^2 \cup W^3$ and $t' \in W^1$, if $c_j \in W^1$, then $W^1$ cannot be a connected set. Hence, $c_j \in W^3$. As $c_j$ is an arbitrary vertex of $C$ in this reasoning, we have $C \subseteq W^3$. Similarly, $B \subseteq W^3$. This implies $C \cup B \subseteq W^3$. By a symmetric argument, we have $C' \cup B' \subseteq W^1$.

We now construct an assignment $\pi: \{X_1, X_2, \dots, X_N\} \mapsto \{\true, \false\}$ using $\calW$. For every $i \in [N]$, set $\pi(X_i) = \true$ if $x_i \in W^1$ (or equivalently $x'_i \in W^0$) and set $\pi(X_i) = \false$ if $x'_i \in W^3$ (or equivalently $x_i \in W^2$).  As mentioned before, $x_i \in W^1$ if and only if $x'_i \in W^0$ and $x'_i \in W^3$ if and only if $x_i \in W^2$. As $W^3$ is connected and $f,c_j \in W^3$, for every $j \in [M]$, there exists $i \in [N]$, such that $x'_i \in W^3$ and $c_jx_i \in E(G)$. Similarly, as $W^1$ is connected, for every $j \in [M]$, there exists $i \in [N]$ such that $x_i \in W^1$ and $c'_jx_i \in E(G)$.  
\end{proof}

\section{Conclusion and Future Directions}
\label{sec:future-direc}
In this work, we showed that \textsc{$C_\ell$-Contractibility} is \NP-complete\ on bipartite graphs for $\ell \in \{4,5\}$ by giving polynomial-time reductions from \textsc{Positive NAE-SAT}. \textsc{Positive NAE-SAT} (or equivalently, \textsc{Hypergraph $2$-Colorability}) has been one of the canonical \NP-complete\ problems in many intractability results on \textsc{$C_\ell$-Contractibility} (\cite{BrouwerV87,DabrowskiP17-IPL,FialaKP13}). 
In general, in most contraction problems, it is a non-trivial task to forbid certain edges from being contracted in a solution.
The simultaneous property of requiring a variable to be \true\ and a variable to be \false\ in every clause of a \yes-instance\ of \textsc{Positive NAE-SAT} helps to encode that certain edges in the output graph of the reduction  cannot be contracted, hence, giving a handle on the required structure of the witness sets.
This is one of the reasons that makes \textsc{Positive NAE-SAT} an amenable choice in many reductions for graph contractibility problems.
However, the sophistication level of the gadgets involved in the reduction increases with the restriction required on the input graph (eg. bipartite graphs, claw-free graphs).
In contrast, the sophistication decreases with increase in the size of the target graph, for instance, the gadgets required for the \NP-hardness\ of \textsc{$C_4$-Contractibility} are more complex than those needed for \textsc{$C_5$-Contractibility}, which are more complex that what are required for \textsc{$C_6$-Contractibility}. 

Continuing along the direction of solving cycle contractibility in restricted graph classes, we can also show the following result.

\begin{restatable}{theorem}{cfourkfourfree}
\label{thm:c4-con-k4-free}
$C_4$-\textsc{Contractibility} is \NP-complete\ on $K_4$-free graphs of diameter 2.
\end{restatable}

We postpone the proof of the theorem in Subsection~\ref{sub-sec:c4-con-k4-free}.
In the subsection, we also argue that Theorem \ref{thm:c4-con-k4-free} can be generalized to show that $K_{p,q}$-\textsc{Contractibility} (the problem of determining if a graph is contractible to the complete bipartite graph with $p$ vertices in one part and $q$ vertices in the other part) is also \NP-complete\ for each $p,q \geq 2$ on $K_4$-free graphs of diameter $2$. 
Our interest in this restricted case stems from its relationship with \textsc{Disconnected Cut}, the problem of determining if a connected graph $G$ contains a subset $U \subseteq V(G)$ such that both $G[U]$ and $G-U$ are disconnected (\cite{ItoKPT11,MartinP15,MartinPL20}). \cite{ItoKPT11} showed that if the diameter of $G$ is $2$, then $G$ has a disconnected cut if and only if $G$ is contractible to $K_{p, q}$ for some $p, q \ge 2$. 
\cite{MartinPL20} proved that \textsc{Disconnected Cut} is polynomial-time solvable for $H$-free graphs when $H \neq K_4$ is a graph on at most $4$ vertices. 
Theorem~\ref{thm:c4-con-k4-free} (and its generalization to $p,q \geq 2$) implies that $(p, q)$-\textsc{Cut} (see \cite{ItoKPT11} for the definition) is \NP-complete\ for all $p, q \ge 2$ on $K_4$-free graphs.
Although this falls short of completing the dichotomy result of \cite{MartinPL20}, we believe that it strongly suggests that there is no polynomial-time algorithm for \textsc{Disconnected Cut} on $K_4$-free graphs.

Finally, determining the longest cycle to which an $H$-free graph (for a fixed $H$) is contractible is another interesting future direction. \cite{KernP20} studied $H$-free graphs in the context of longest paths.
Note that assuming \P$\neq$\NP, the complexities of contracting to a longest path and longest cycle do not coincide on $H$-free graphs.

\subsection{Proof of Theorem~\ref{thm:c4-con-k4-free}}
\label{sub-sec:c4-con-k4-free}

It is easy to verify that the problem is in \NP. To show \NP-hardness, once again we give a polynomial-time reduction from \textsc{Positive NAE-SAT}. Let $\{X_1, X_2, \dots, X_N\}$ and $\{C_1, C_2, \dots, C_M\}$ be the sets of variables and clauses, respectively, in an instance $\psi$ of \textsc{Positive NAE-SAT}. We construct a $K_4$-free diameter 2 graph $G$ with a partition $\{X,S^+,S^-\}$ of its vertex set as follows. 
\begin{enumerate}
\item Add vertices $t, t'$ to $S^+$ and vertices $f, f'$ to $S^-$. Let $A$ denote $\{t,t',f,f'\}$.
\item For every $i \in [N]$, we add a vertex $x_i$ to $X$ corresponding to variable $X_i$. 
\item For every $j \in [M]$, we add vertices $c^+_j, b^+_j$ to $S^+$ and $c^-_j, b^-_j$ to $S^-$ corresponding to clause $C_j$. 
\item Make every vertex in $S^+$ adjacent to every vertex in $S^-$.
\item Make every vertex in $X$ adjacent to both $t$ and $f$. 
\item For every $i \in [N]$ and $j \in [M]$ such that $X_i$ appears in $C_j$, add edges $x_i c^+_j$, $x_i b^+_j$, $x_i c^-_j$, and $x_i b^-_j$. 
\end{enumerate}

This completes the construction of $G$. See Figure \ref{fig:c4-con-k4-free} for an illustration. It is easy to verify that $G$ is $K_4$-free as the three sets $S^+$, $S^-$, $X$ that partition $V(G)$ are independent sets. Further, the diameter of $G$ is two as any pair of non-adjacent vertices have a common neighbor. We now prove the correctness of the reduction. Suppose $\pi : \{X_1, X_2, \dots, X_N\} \mapsto \{\true, \false\}$ is a not-all-equal satisfying assignment of $\psi$. Then, let $W^0 = \{t'\}$, $W^1 = \{f'\}$, $W^2 = \{t\} \cup \{x_i \mid i \in [N]\ \text{and}\ \pi(X_i) = \true\} \cup (S^+\setminus \{t'\})$ and $W^3 = \{t\} \cup \{x_i \mid i \in [N]\ \text{and}\ \pi(X_i) = \false\} \cup (S^-\setminus \{f'\})$. It is easy to verify that $\{ W^0, W^1, W^2, W^3 \}$ is a $C_4$-witness structure of $G$. 

\begin{figure}[t]
  \begin{center}
    \includegraphics[scale=1]{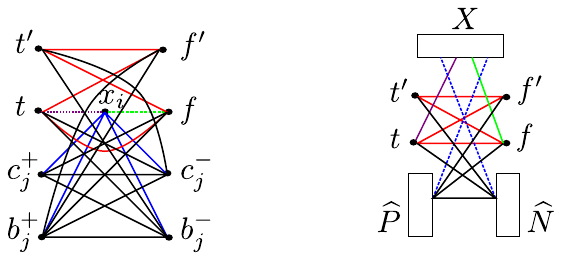}
    \end{center}
   \caption{(Left) The graph $G$ with different sets of edges highlighted. Here, purple (dotted) edges denote setting the variable to \true\ and the green (dashed) edges denote setting it to \false\ in a \yes-instance. (Right) Adjacency relation between subset of vertices where $\widehat{P}=S^+ \setminus \{t,t'\}$ and $\widehat{N}=S^- \setminus \{f,f'\}$. 
    \label{fig:c4-con-k4-free}}
\end{figure}

Conversely, suppose $\calW = \{W^0, W^1, W^2, W^3\}$ is a $C_4$-witness structure of $G$ where $E(W^i, W^j) \neq \emptyset$ if and only if $j=(i+1) \mod 4$. Then, we claim that every pair of vertices in $A$ are in different witness sets. Assume that the claim holds. Then, without loss of generality, we may assume that $t' \in W^0$, $f' \in W^1$, $t \in W^2$, and $f \in W^3$. Observe that as $f$ and $t$ are adjacent to every vertex in $X$, it follows that for each $i \in [N]$, $x_i \in W^2 \cup W^3$. We now construct an assignment $\pi: \{X_1, X_2, \dots, X_N\} \mapsto \{\true, \false\}$ using $\calW$. For every $i \in [N]$, set $\pi(X_i) = \true$ if $x_i \in W^2$ and set $\pi(X_i) = \false$ if $x_i \in W^3$. As every vertex in $S^+$ is adjacent to $f$ and $f'$, for every $j \in [M]$, we have $c^+_j,b^+_j \in W^0 \cup W^2$. Similarly, as every vertex in $S^-$ is adjacent to $t$ and $t'$, for every $j \in [M]$, we have $c^-_j,b^-_j \in W^1 \cup W^3$. If for some $j \in [M]$, $c^+_j \in W^0$, then as $W^0$ has no neighbors of  $c^+_j$, $W^0$ cannot be a connected set. Then, as $N(b^+_j)=N(c^+_j)$, we have $S^+ \subseteq W^2$. Similarly, we have $S^- \subseteq W^3$. Consider $j \in [M]$. As $t,c^+_j,b^+_j \in W^2$ and these three vertices form an independent set, it follows that there is an index $i \in [N]$ such that $x_i \in W^2$ satisfying $x_ic^+_j,x_ib^+_j \in E(G)$. Similarly, for each $j \in [M]$, there is an index $i \in [N]$ such that $x_i \in W^3$ satisfying $x_ic^-_j,x_ib^-_j \in E(G)$. It now follows that for any $j \in [M]$, $\pi$ sets at least one of the variables in clause $C_j$ to \true\ and at least one of variables in $C_j$ to \false. 

It now remains to show that in the $C_4$-witness structure $\calW$ of $G$, every pair of vertices in $A$ are in different witness sets. As $S_1=\{t,f\}$, $S_2=\{t',f\}$ and $S_3=\{t,f'\}$ are dominating sets in $G$, for each $i \in [3]$, the vertices in $S_i$ are in different witness sets. If for some $i$, the two vertices $S_i$ are in the same witness set $W^j$, then it follows that $W^{(j+2) \mod 4}$ is empty leading to a contradiction. Now suppose $f$ and $f'$ are in the same witness set, say $W^0$. Then, $W^2 \subseteq S^- \setminus \{f,f'\}$ as every other vertex is adjacent to $f$ or $f'$. However, as $S^- \setminus \{f,f'\}$ is an independent set, it follows that $W^2$ is a singleton set, say $\{c^-_j\}$. Then, $N(c^-_j) \cap W^1 \neq \emptyset$ and  $N(c^-_j) \cap W^3 \neq \emptyset$. As $N(c^-_j)=N(b^-_j)$, we have $b^-_j \in W^0$. However, $b^-_j, f,f' \in W^0$ implies that $W^0$ is not a connected set (as all neighbors of $b^-_j$ are in $W^1 \cup W^3$) leading to a contradiction. A similar argument holds if $W^2 = \{b^-_j\}$. A similar argument shows that $t$ and $t'$ cannot be in the same witness set as well. Finally, we show that $t'$ and $f'$ are in different witness sets. Assume on the contrary that $t'$ and $f'$ are in the same witness set, say $W^0$. Then, as $t$ and $f$ cannot be in $W^0$, we have $t,f \in W^1 \cup W^3$. However, as $t$ and $f$ are in different witness sets, it follows that $E(W^1,W^3)\neq \emptyset$ leading to a contradiction. \\

\noindent {\bf Remark:} Observe that the base cycle $(t',f',t,f)$ in the above construction may be viewed as a $K_{2,2}$. Theorem \ref{fig:c4-con-k4-free} can be generalized to show that $K_{p,q}$-\textsc{Contractibility} is \NP-complete\ for each $p,q \geq 2$ on $K_4$-free graphs of diameter $2$ by blowing up the graph $G$ as follows: add $p-2$ new vertices to $S^+$ and $q-2$ vertices to $S^-$ ensuring that every vertex in $S^+$ is adjacent to every vertex in $S^-$. Now, we can show that $\psi$ is a \yes-instance\ of \textsc{Positive NAE-SAT} if and only if $G$ is contractible to $K_{p,q}$.

\nocite{*}
\bibliographystyle{abbrvnat}
\bibliography{refs}
\label{sec:biblio}

\end{document}